\documentclass[journal]{IEEEtran}
\usepackage{cite}
\usepackage{color}
\usepackage{graphicx}
\usepackage{amsmath}
\usepackage{amsfonts}
\usepackage{amsthm}
\usepackage{bm}
\theoremstyle{plain}

\newtheorem{lem}{Lemma}

\newtheorem{prop}{Proposition}

\newtheorem{remark}{Remark}

\usepackage{algorithm}
\usepackage{multirow}  
\usepackage{booktabs}  
\usepackage{algpseudocode}

\usepackage{amsfonts}
\usepackage{amsthm}
\usepackage{amsmath}%
\usepackage{MnSymbol}%
\usepackage{wasysym}%

\ifCLASSOPTIONcompsoc
  \usepackage[caption=false,font=normalsize,labelfont=sf,textfont=sf]{subfig}
\else
  \usepackage[caption=false,font=footnotesize]{subfig}
\fi

\usepackage{stfloats}
\renewcommand{\arraystretch}{1.3}

\usepackage{nomencl}
\usepackage{ifthen}
\renewcommand{\nomgroup}[1]{%
    \ifthenelse{\equal{#1}{A}}{\item[\emph{\textbf{Set and Index}}]}{%
    \ifthenelse{\equal{#1}{B}}{\item[\emph{\textbf{Parameters}}]}{%
    \ifthenelse{\equal{#1}{C}}{\item[\emph{\textbf{Variables}}]}
    }
    }
    }
\makenomenclature

\begin{document}

\title{Proactive Robust Hardening of Resilient Power Distribution Network: Decision-Dependent Uncertainty Modeling and Fast Solution Strategy}

\author{Donglai Ma,
	        Xiaoyu Cao,~\IEEEmembership{Member,~IEEE}, 
	        Bo Zeng,~\IEEEmembership{Member,~IEEE},
	        Qing-Shan Jia,~\IEEEmembership{Senior Member,~IEEE}, \\
	        Chen Chen,~\IEEEmembership{Senior Member,~IEEE}, 
	        Qiaozhu Zhai,~\IEEEmembership{Member,~IEEE},     
	        and Xiaohong Guan,~\IEEEmembership{Life Fellow,~IEEE} 
	
	\thanks{This work was partially supported by the National Key Research and Development Program of China under Grant 2022YFA1004600, and by the National Natural Science Foundation of China under Grant 62373294 and 62192752. \emph{(Corresponding author: Xiaoyu Cao.)}}%
	\thanks{Donglai Ma, Xiaoyu Cao and Qiaozhu Zhai are with the School of Automation Science and Engineering and the Ministry of Education Key Laboratory for Intelligent Networks and Network Security, Xi’an Jiaotong University, Xi’an 710049, Shaanxi, China (e-mail: mdl1622267691@stu.xjtu.edu.cn; cxykeven2019@xjtu.edu.cn; qzzhai@sei.xjtu.edu.cn).}%
	\thanks{Bo Zeng is with the Department of Industrial Engineering and the Department of Electrical and Computer Engineering, University of Pittsburgh, Pittsburgh, PA 15106 USA (e-mail: bzeng@pitt.edu).}%
	\thanks{Qing-Shan Jia is with the Center for Intelligent and Networked Systems, Department of Automation, Tsinghua University, Beijing 100084, China (e-mail: jiaqs@mail.tsinghua.edu.cn).}%
	\thanks{Chen Chen is with the School of Electrical Engineering, Xi’an Jiaotong University, Xi’an 710049, Shaanxi, China (e-mail: morningchen@xjtu.edu.cn).}%
	\thanks{Xiaohong Guan is with the School of Automation Science and Engineering and the Ministry of Education Key Laboratory for Intelligent Networks and Network Security, Xi’an Jiaotong University, Xi’an 710049, Shaanxi, China, and also with the Center for Intelligent and Networked Systems, Department of Automation, Tsinghua University, Beijing 100084, China (e-mail: xhguan@xjtu.edu.cn).}%
	}

\maketitle

\begin{abstract}
To address the power system hardening problem, traditional approaches often adopt robust optimization (RO) that considers a fixed set of concerned contingencies, regardless of the fact that hardening some components actually renders relevant contingencies impractical. In this paper, we directly adopt a dynamic uncertainty set that explicitly incorporates the impact of hardening decisions on the worst-case contingencies, which leads to a decision-dependent uncertainty (DDU) set.  Then, a DDU-based robust-stochastic optimization (DDU-RSO) model is proposed to support the hardening decisions on distribution lines and distributed generators (DGs). Also, the randomness of load variations and available storage levels is considered through stochastic programming (SP) in the innermost level problem. Various corrective measures (e.g., the joint scheduling of DGs and energy storage) are included, coupling with a finite support of stochastic scenarios, for resilience enhancement. To relieve the computation burden of this new hardening formulation, an enhanced customization of parametric column-and-constraint generation (P-C\&CG) algorithm is developed. By leveraging the network structural information, the enhancement strategies based on \textit{resilience importance indices} are designed to improve the convergence performance. Numerical results on 33-bus and 118-bus test distribution networks have demonstrated the effectiveness of DDU-RSO aided hardening scheme. Furthermore, in comparison to existing solution methods, the enhanced P-C\&CG has achieved a superior performance by reducing the solution time by a few orders of magnitudes. 
\end{abstract}

\def\abstractname{Note to Practitioners}
\begin{abstract}
This paper proposes an optimal hardening model for enhancing the resilience of power distribution systems. Our planning model is extended under a trilevel RSO framework incorporating the DDU set. Compared to the traditional RO built upon the static uncertainty sets, the DDU-RSO formulation has a dynamic and stronger modeling capability, which captures the endogenous nature of contingency uncertainties associated with proactive hardening decisions. Also, the exogenous uncertainty of power loads and available storage levels is considered by introducing a scenario-based SP formulation in the innermost level problem. Note that the real shape and size of DDU set are changeable depending on the proactive decisions, which makes our problem very challenging to compute. To exactly and efficiently solve the DDU-RSO model, we design and implement a customized P-C\&CG algorithm. Also, the enhancement strategy based on \textit{resilience importance indices} is applied, which helps address the multi-solution dilemma widely existing in network optimization problems. It is worthy noting that the resilience importance indices can be attained not only using numerical calculation, but also through the prior knowledge of system operators. The probability information of contingency events is also helpful for determining the resilience importance indices to reduce the conservatism of RO decisions. By exploiting the ``deep knowledge" contained in power systems, the customized P-C\&CG algorithm has demonstrated a superior scalable capability over existing methods (e.g., the basic C\&CG and a type of Benders decomposition algorithm). With this advanced analytical tool, the proposed optimization framework can be readily implemented for resilient planning and operation of real-world power systems.
\end{abstract}


\IEEEpeerreviewmaketitle

\section{Introduction}

\IEEEPARstart{T}{he} high-impact and low-probability extreme weather events (e.g., hurricanes, floods and ice storms) may seriously destroy the electrical power infrastructure, especially the power distribution network (PDN) \cite{poudyal2022risk, chen2017modernizing}. 
For preparing the PDN to be intrinsically resilient against hazardous damages, the two-stage robust optimization (RO) method has been heavily adopted in existing studies (e.g., \cite{mishra2021review,shi2022enhancing,hasanzad2021robust,yuan2016robust}) to support the risk-averse hardening decisions. One evident advantage of RO is its weak reliance on empirical data (which could be difficult to attain due to the low occurrence frequency of hazardous events) to construct exact probability distributions or a scenario tree. Instead, it seeks improved resilience out of an enveloped uncertainty set under the worst-case consideration. 

One central component behind implementing the RO-based hardening method lies on the appropriate modeling of contingency uncertainties, i.e., defining the uncertainty set. 
Ref \cite{yuan2016robust} introduced a two-stage robust optimization (RO) model that integrates line hardening and the allocation of distributed energy resources (DERs) within a multi-stage, multi-zone uncertainty set designed to capture the spatial and temporal dynamics of extreme weather events. Building on this, Ref \cite{Zhang2021TSG} extended the model to include additional uncertainties, e.g., the speed and direction of typhoons. Beyond just DERs, strategies for enhancing network resilience also involve coordinating line hardening with other measures, e.g., the topology reconfiguration realized by remotely controlled smart-switches \cite{lin2018tri, li2023resilience}, and the scheduling of mobile energy resources \cite{tao2022multi}. Furthermore, the resilient hardening strategies for coupled power distribution and transportation systems are also studied in some literature, e.g., in Ref \cite{wang2018resilience, gan2021tri}.
In the conventional RO framework, the uncertainty set is generally static, predetermined by prior knowledge (e.g., the $N-k$ contingency). However, this static approach may not adequately reflect the dynamic changes in the system induced by proactive hardening actions. For example, hardening measures such as overhead structure reinforcement, vegetation management, undergrounding of power lines, and elevating substations or distributed generators (DGs) can greatly reduce the damage risks of network assets in response to extreme weather events \cite{mishra2021review, Zhang2021TSG, 7514755, lin2018tri}.  Specifically, reinforcing overhead structures is a fundamental approach to harden the distribution systems, including upgrading towers and poles with strengthened materials, enhancing guying, and increasing the number of supporting poles. Effective vegetation management is also crucial, as fallen trees are one of the major factors to cause power outages under severe storms \cite{7514755}. Burying distribution lines as underground cables can significantly decrease the system's vulnerability to strong winds, lightning, and vegetation interference. But it could raise very high capital expenditures (CaPEX). Some other measures, e.g., elevating substations or DGs, and relocating critical facilities to more safety areas, could help mitigate the disruption of floods \cite{movahednia2021power,souto2022power}.
So there is a strong correlation between hardening decisions and the corresponding contingencies. Nevertheless, this relationship is largely overlooked in the existing literature, with a preference for static or decision-independent uncertainty (DIU) sets. It is clear that the uncertainty of PDN concerning hardening decisions evidently exhibits endogenous nature, also known as the \emph{decision-dependent uncertainty} (DDU) \cite{zeng2022two,Zhang2022,li2023distributionally}.

Although a few existing studies on PDN hardening consider the impact of proactive decisions on contingencies when construct the uncertainty set, 
the formulations are solved heuristically, or develop a DIU-based model and then computed using traditional method. 
Ref \cite{Wang2019} proposed a DIU-RO model that improved the resilience of PDN with controllable and renewable distributed generators by coordinating with line hardening and dynamic microgrids formation. In Ref \cite{shen2023resilience}, the network hardening and DERs configuration were considered in long-term preventive investment for resilience purposes. Ref \cite{9963579} presented a standard RO model to optimize both the hardening of tie lines and the deployment of bilateral tie switches. Also, a progressive detection mechanism was designed for estimating the potential propagation of power outages.
To improve the integrated-energy resilience, Ref \cite{li2023robust} developed a robust expansion planning and hardening framework for a multi-energy distribution system with interdependent electricity, natural gas and heating networks. Ref \cite{he2018robust} studied the proactive robust hardening of integrated electricity-gas network. The Claude Shannon’s information theory was adopted to formulate the uncertainty set instead of the classical $N-k$ criterion.  
Such type of uncertainty set was also utilized in \cite{7514755} to characterize the impact of different hardening strategies on network contingency probability. 

Actually, we have noted some recent studies on DDU-based RO methodology with power system applications. A two-stage robust generation dispatch model was proposed in Refs \cite{chen2022robust,tan2024robust}, which determined the preparatory curtailment of renewable energy (RE) generation in the first stage problem. The curtailment threshold could affect the fluctuation range of real-time RE outputs, leading to a DDU set. Considering the uncertainty of power loads and its correlation with time-of-use (ToU) prices, Ref \cite{10122720} developed a DDU-based robust unit commitment (UC) model to support the generation scheduling with a well-designed ToU pricing scheme. The hybrid modeling of endogenous and exogenous uncertainties was considered in Ref \cite{haghighat2023robust} for capacity sizing of microgrids. Ref \cite{zhang2021robust} studied a virtual power plant (VPP) scheduling problem, which addressed the DIUs associated with market clearing prices and wind power generation, as well as the DDUs pertaining to real-time reserve demanding. In Ref \cite{gu2023distributed}, a RO-based multi-port and multi-period feasible region formulation was developed to capture the flexibility of distribution network operations. The uncertainty of dispatch instructions was modeled as a DDU set, and the worst-case scenario was identified for operational security assessment. On the other hand, a variety of solution algorithms were designed and implemented for computing the DDU-based RO models. Ref \cite{chen2023robust} proposed an alternating direction algorithm and a column generation (CG) algorithm to compute the RO problem with continuous form of DDU set. For the RO problems with mixed-integer DDU sets, the modified column-and-constraint generation (C\&CG) in \cite{10541096} and the adaptive C\&CG in \cite{chen2022robust} were developed by leveraging special problem structures. Although the dual C\&CG approach proposed in \cite{tan2024robust} provided a more general measure for handling mixed-integer variables in the DDU set, it may expand the feasible region of the uncertainty set, potentially reducing the solution efficiency. In contrast, the modified Benders decomposition (BD) algorithm proposed in \cite{Zhang2022} and the parametric C\&CG (P-C\&CG) algorithm proposed in \cite{zeng2022two} are general algorithms to exactly  solve DDU-based RO formulations. Moreover, as indicated by Ref \cite{zeng2022two} and the numerical results of our study, the P-C\&CG  demonstrates a superior solution performance compared to the BD algorithms, especially for large-scale problems.	

However, most of the current resilience studies (e.g., Refs \cite{yuan2016robust,Zhang2021TSG,lin2018tri, li2023resilience,tao2022multi,7514755, 10541096}) only consider the contingency uncertainty, while neglecting other critical uncertain factors (e.g., the random variation of power loads). Moreover, the practical correlation between hardening decisions and the corresponding contingencies is often overlooked in the existing literature, so they adopt the static or DIU sets. Therefore, analytical studies and deep insights into the holistic representation of endogenous and exogenous uncertainties for resilient system hardening are still lacking. Indeed, improper characterization of the structural properties of the DDU set may influence its robust feasibility, resulting in overly conservative outcomes.

In this paper, we present a resilient PDN hardening method under endogenous and exogenous uncertainties. Note that the endogenous uncertainty lies on the contingency of network assets, e.g., distribution lines and DGs. A closed-form DDU set is developed to represent the contingency uncertainty as a point-to-set mapping of proactive hardening decisions. Then, the exogenous uncertainties are considered as the randomness of power loads and available storage levels (at the beginning of service restoration), which also could largely influence the load shedding severity. To capture the exogenous uncertainties, a finite support of stochastic scenarios is introduced associating with the operational corrective decisions. As a consequence, a novel DDU-based robust-stochastic optimization (DDU-RSO) formulation is proposed. Specifically, a trilevel program following “defender-attacker-defender” framework is developed to optimize the hardening decisions, subject to the worst-case contingency scenario under DDU perspectives. Moreover, the innermost level problem is modeled as a scenario-based stochastic program (SP). Various operational corrective measures (e.g., the joint scheduling of DGs and energy storage) are coordinated with respect to stochastic scenarios, aiming to minimize the mathematical expectation of load curtailment.
	
Different from conventional DIU-based RO, the real shape and size of DDU set are dynamically changeable with the proactive decision, making our problem rather challenging to compute. To relieve the computation burdens, an enhanced customization of P-C\&CG algorithm is developed. By enumerating the Cartesian product of the vertex sets in the recourse dual feasible regions across all scenarios, and constructing the corresponding optimality conditions, the enhanced P-C\&CG can achieve exact and efficient solution. Compared to the P-C\&CG proposed in \cite{zeng2022two}, our algorithm is specifically designed to address the DDU-RSO formulation (handling a mixture of endogenous and exogenous uncertainties). By exploiting the structural information of distribution networks, the \emph{resilience importance indices} are utilized to improve the convergence performance. Although the basic idea of  ``resilience importance indices'' has appeared in the existing publications \cite{shen2012discovery,fan2022critical,zhang2020critical} to heuristically study the critical components of a grid, to the best of our knowledge, it is the first time that this concept is applied to enhance an exact solution algorithm on this regard. As demonstrated by our numerical results, the proposed algorithm outperforms the basic C\&CG, and is drastically superior to the Benders decomposition (BD) algorithm. Particularly, compared to the modified BD in \cite{Zhang2022,zhang2021robust}, the enhanced P-C\&CG can solve the DDU-RSO formulation by a few orders of magnitude faster.  

In comparison to current studies, the major contributions of this paper are summarized as below:
\begin{enumerate}
\item The endogenous nature of contingency uncertainties for a proactively reinforced PDN, which considers both the disruptions of distribution lines and DGs, is modeled and analytically represented as a closed-form DDU set.  
\item A trilevel RSO formulation with DDU set is proposed to support the resilient hardening decisions considering both endogenous and exogenous uncertainties. The latter ones are modeled as a scenario-based SP in the lower-level operational problem.
\item To address the computational challenges of DDU-RSO model, a customized P-C\&CG algorithm is designed and enhanced by fully exploiting the problem structure.
\end{enumerate}

The remainder of this paper is organized as below. Section \ref{DDU-RSO} formulates the DDU-RSO model. Section \ref{PCCG} develops an enhanced customization of P-C\&CG algorithm. Section \ref{Numerical} shows the numerical results by conducting comparative case studies. The conclusions are presented and discussed in Section \ref{Conclude}.

\section{DDU-RSO Based Resilient Hardening Model}
\label{DDU-RSO}
We consider a two-stage optimization framework for the hardening decisions of PDN. The tree-spanning topology of a PDN can be represented by sets of nodes and branches (i.e., $\mathcal{N}$ and $\mathcal{B}$). Note that in the first stage, the proactive hardening variables $\bm{x}$ ($= \bm{z} \bigcup \bm{r}$, a binary vector indicating the reinforcement states of distribution lines ($\bm{z}$) and DGs ($\bm{r}$)) should be determined before the realization of any contingency, which is represented by a binary vector $\bm{u}$ ($= \bm{\omega} \bigcup \bm{\nu}$, including the faulting states of lines ($\bm{\omega}$) and DGs ($\bm{\nu}$). Then, in the second-stage (recourse) problem, the joint scheduling of DGs and ESSs is considered as the operational corrective measures. To capture the randomness of power loads variation as well as the available storage levels at the starting period of service restoration, we introduce a finite support of stochastic scenarios $\mathcal{S}=\left\lbrace \xi_s| s=1,2,\ldots,N_S \right\rbrace $. Parameter $\pi_s$ is to denote the probability of scenario $s$, such that $\sum_{s\in \mathcal{S}} \pi_s =1$. As a consequence, the innermost-level problem can be modeled as a scenario-based SP, where the post-event corrective decision $\bm{y}_s$ can be settled to minimize the mathematical expectation of load shedding.

Let $\mathcal{U}(\bm{x})$  denote the decision-dependent contingency set, a point to set mapping for $\bm{u}$. We next present the DDU-based robust hardening problem with innermost stochastic representation ($\mathbf{DDU-RSO}$) as follows:  
\begin{equation}
	\label{OBJ}
	\Gamma=\min_{\bm{x}\in \mathcal{X}} \max_{\bm{u}\in \mathcal{U}(\bm{x})} \sum_{s\in \mathcal{S}} \pi_s \min_{\bm{y}_s \in \mathcal{Q}_s(\bm{u})}   \sum_{j\in \mathcal{N}}\sum_{t\in \mathcal{T}} \rho_j p^{c,s}_{j,t}
\end{equation}
where 
\begin{equation}
	\label{FS}
	\mathcal{X}=\left\lbrace \bm{x}\in \{0,1\}^{{n_1}\times{n_2}}: \ \sum_{(i,j)\in \mathcal{B}} c_{ij}^L z_{ij} +\sum_{j\in \mathcal{N}^{dg}} c_{j}^G r_{j} \leq b \right\rbrace 
\end{equation}
with $z_{ij}$/$r_j$ being the binary hardening decisions on distribution line $(i,j)$/DG $j$, where $z_{ij}=1$/$r_{j}=1$ when it is structurally protected and $0$ otherwise. Parameters $c_{ij}^L$/$c_j^G$ and $b$ represent the hardening cost factor of line $(i,j)$/DG $j$ and the cost budget for the entire hardening plan, respectively. Also, we have $p^{c,s}_{j,t}$ being the variables of active load shedding at time slot $t$ in scenario $s$, where the priority weight of loads serving at node $j$ is denoted by $\rho_j$.

We note that burying the overhead power lines as underground cables and constructing barriers for DGs can effectively protect them from extreme weather events. Because the fault probability can be reduced to an extremely low level by these hardening measures, they are often assumed to be fully protected under catastrophic occasions \cite{Zhang2021TSG,gan2021tri,tao2022multi}. The contingency uncertainty set associating with proactive hardening decisions $\bm{x}$ can be written as:
\setlength{\arraycolsep}{-0.3em}
\begin{eqnarray}
	&&\mathcal{U}(\bm{x})=\{ \bm{u}\in \{0,1\}^{{m}_1\times{m}_2}:  \nonumber \\
	&&\quad \quad \quad \quad \quad \quad \quad \sum_{(i,j)\in \mathcal{B}^{v}} (1-\omega_{ij}) \leq k^L, \label{DDU:CONS1}  \\
	&&\quad \quad \quad \quad \quad \quad \quad \omega_{ij} \geq z_{ij}, \ \forall (i,j) \in \mathcal{B}^v, \label{DDU:CONS2} \\
	&&\quad \quad \quad \quad \quad \quad \quad \sum_{j\in \mathcal{N}^{vdg}} (1-\nu_j) \leq k^G, \label{DDU:CONS3}  \\
	&&\quad \quad \quad \quad \quad \quad \quad  \nu_{j} \geq r_{j}, \ \forall j \in \mathcal{N}^{vdg} \label{DDU:CONS4} \}
\end{eqnarray}
where $\mathcal{B}$ and $\mathcal{N}^{dg}$ represent the set of distribution lines and DGs, respectively, while $\mathcal{B}^{v}$ and $\mathcal{N}^{vdg}$ denote the corresponding vulnerable sets. Note that the vulnerable sets are built to capture practical factors (e.g., disasters pre-warning information, geographical environment, and the aging status of  network facilities). The system components excluded from the vulnerable sets (i.e., $\mathcal{B} \backslash \mathcal{B}^v$ and $\mathcal{N}^{dg} \backslash \mathcal{N}^{vdg}$) are assumed to be free of the disastrous influence. So their operational constraints always hold, and can be omitted from the DDU set. Binary variables $z_{ij}$ and $r_j$ represent the hardening status of line $(i,j)$ and generator $j$, respectively. Uncertain variables $\omega_{ij}$/$\nu_j$ denote the contingency status of line $(i,j)$/DG $j$, where $\omega_{ij}=0$/$\nu_j=0$ if it is damaged and $1$ otherwise. Following the $N-k$ criterion, the total number of faulting lines and DGs are constrained by \eqref{DDU:CONS1} and \eqref{DDU:CONS3}, respectively. Besides, the proactive hardening decisions impose restrictions on $\omega_{ij}$/$\nu_j$, reflecting a decision-dependent relationship as in \eqref{DDU:CONS2} and \eqref{DDU:CONS4}. Taking \eqref{DDU:CONS1}-\eqref{DDU:CONS2} as an example, as long as $z_{ij}=1$, power line $(i,j)$ is assumed to be under a complete protection (e.g., by overhead lines with underground cables), so that $\omega_{ij}$ will be fixed to 1. In other instances, however, they are free variables. Similar assumptions are adopted for the DDU model in \eqref{DDU:CONS3}-\eqref{DDU:CONS4}.

Moreover, the feasible set of second-stage recourse decisions with respect to each stochastic scenario $s \in \mathcal{S}$ (subject to hardening actions and a DDU set) is presented as below:
\setlength{\arraycolsep}{-0.0em}
\begin{eqnarray}
	&& \bm{y}_s=\left\lbrace p_{j,t}^s, q_{j,t}^s,p_{j,t}^{c,s},q_{j,t}^{c,s},U_{j,t}^s|\forall j\in \mathcal{N},\forall t\in \mathcal{T}  \right\rbrace \nonumber \\ 
	&&\quad \quad  \bigcup \left\lbrace P_{ij,t}^s,Q_{ij,t}^s|\forall (i,j)\in {\cal B},\forall t\in \mathcal{T} \right\rbrace \subseteq \mathbb{R}^{e}\nonumber\\ 
	&& \mathcal{Q}_s(\bm{u})=\{\bm{y}_s\in \mathbb{R}^{e}: \nonumber \\ 
	&& \sum\limits_{i \in \psi (j)} {{P_{ij,t}^s} - } \sum\limits_{k \in \phi (j)} {P_{jk,t}^s}  = PD_{j,t}^s-p_{j,t}^{c,s}-p_{j,t}^s-gp_{j,t}^s,  \nonumber\\
	&&\quad \quad \quad \quad\quad \quad\quad \quad\quad \quad \quad\quad \quad\quad \quad\quad\quad \forall j \in {{\cal N}} ,\forall t\in {{\cal T}}  \label{RO:CONS1} \\
	&& \sum\limits_{i \in \psi (j)} {{Q_{ij,t}^s} - } \sum\limits_{k \in \phi (j)} {{Q_{jk,t}^s}}  = QD_{j,t}^s-q_{j,t}^{c,s}-q_{j,t}^s-gq_{j,t}^s, \nonumber\\
	&&\quad \quad \quad \quad\quad \quad\quad \quad\quad \quad \quad\quad \quad\quad \quad\quad\quad \forall j \in {{\cal N}} ,\forall t\in {{\cal T}}  \label{RO:CONS2} \\
	&& 0 \le p_{j,t}^{c,s} \le PD_{j,t}^s, \quad \forall j \in {{\cal N}} ,\forall t\in {{\cal T}}\label{RO:CONS3} \\
	&& 0 \le q_{j,t}^{c,s} \le QD_{j,t}^s, \quad \forall j \in {{\cal N}} ,\forall t\in {{\cal T}}\label{RO:CONS4} \\
	&& \nu_j \underline{G}_j \le p_{j,t}^s \le \nu_j \overline{G}_j, \quad \forall j \in {{\cal N}^{dg}},\forall t\in {{\cal T}}\label{RO:CONS5} \\
	&&\nu_j \overline{G}_j \tan\underline{\theta}_j \le q_{j,t}^s \le \nu_j \overline{G}_j \tan\overline{\theta}_j, \quad \forall j \in {{\cal N}^{dg}},\forall t\in {{\cal T}}\label{RO:CONS6} \\
	&& 0 \leq gp_{j,t}^{s} \leq \overline{P}_{j}^{ess}, \quad \forall j \in {{\cal N}^{ess}},\forall t\in {{\cal T}} \label{RO:CONS15-1} \\
	&& 0 \leq gq_{j,t}^{s} \leq \overline{Q}_{j}^{ess}, \quad \forall j \in {{\cal N}^{ess}},\forall t\in {{\cal T}} \label{RO:CONS15-2} \\
	&& \tilde{E}_j^{s}-\sum_{\tau=0}^t \left(gp_{j,\tau}^{s}/\eta_j^{dis}\right)  \geq 0, \quad \forall j \in {{\cal N}^{ess}},\forall t\in {{\cal T}} \label{RO:CONS16} \\ 
	&&\underline{U}_j \le U_{j,t}^s \le  \overline{U}_j, \quad \forall j \in {{\cal N}},\forall t\in {{\cal T}}\label{RO:CONS10}\\
	&& -\omega_{ij}\overline{P}_{ij} \le P_{ij,t}^s \le \omega_{ij}\overline{P}_{ij}, \quad \forall (i,j)\in {{\cal B}},\forall t\in {{\cal T}}\label{RO:CONS7} \\
	&& -\omega_{ij}\overline{Q}_{ij} \le Q_{ij,t}^s \le \omega_{ij}\overline{Q}_{ij}, \quad \forall (i,j)\in {{\cal B}},\forall t\in {{\cal T}}\label{RO:CONS8} \\
	&&(U_{i,t}^s-U_{j,t}^s-\frac{r_{ij}P_{ij,t}^s+x_{ij}Q_{ij,t}^s}{U_0})\omega_{ij}=0, \forall (i,j) \in \mathcal{B},\forall t\in {{\cal T}}\} \nonumber\\
	&& \label{RO:CONS9} 
\end{eqnarray}
where $PD_{j,t}^s$/$QD_{j,t}^s$ denote the nominal values of active/reactive power loads, respectively. As in power balance equations \eqref{RO:CONS1} and \eqref{RO:CONS2}, the DGs and energy storage systems (ESSs) are employed to fulfill the power demands. The active/reactive power outputs of DGs at node $j$ (belonging to a node subset ${{\cal N}^{dg}}$) are denoted by variables $p_{j,t}^s$/$q_{j,t}^s$. Once the power loads cannot be met, the on-emergency mechanism of intentional load shedding would be activated, as represented by variables $p_{j,t}^{c,s}$/$q_{j,t}^{c,s}$. The shedding loads are required to be nonnegative, as in \eqref{RO:CONS3} and \eqref{RO:CONS4}. The active/reactive power generation of DGs at node $j$  (belonging to a node subset $\mathcal{N}^{dg}$) is constrained by \eqref{RO:CONS5} and \eqref{RO:CONS6}. Parameters $\overline{G}_j$/$\underline{G}_j$ denote the upper/lower bound of active power outputs of DGs. The reactive power served by DGs can be realized using the power electronic equipment, which is adjustable depending on their rated capacities and power factor angles (ranging between $[\underline{\theta}_j, \overline{\theta}_j]$) \cite{cao2022resilience}. Besides, the DGs outputs are restricted by faulting status $\nu_j$. For each ESS unit deployed at node $j$  (belonging to a node subset $\mathcal{N}^{ess}$), its discharging active/reactive power is constrained by the power rating $\overline{P}_{j}^{ess}$/$\overline{Q}_{j}^{ess}$, as in \eqref{RO:CONS15-1} and \eqref{RO:CONS15-2}. Also, the accumulation of discharged energy cannot exceed the available storage level at the beginning of service restoration, as defined in \eqref{RO:CONS16}. Parameter $\eta_j^{dis}$ represents the discharge efficiency of ESSs. Note that the initial storage level is considered as an uncertain variable with respect to scenario $s$, as denoted by $\tilde{E}_j^{s}$.

Additionally, the linear \emph{DistFlow} model \cite{baran1989reconfig} is applied and modified to describe the operations of PND with damaged lines. Variables $P_{ij,t}^s$/$Q_{ij,t}^s$ represent the active/reactive power flows on line $(i,j)$, and $U_{j,t}^s$ denotes the voltage level of node $j$. These network flow variables are constrained by \eqref{RO:CONS7}-\eqref{RO:CONS9}. The allowable fluctuation range of voltage levels is imposed as $\overline{U}_{j}$/$\underline{U}_{j}$ in \eqref{RO:CONS10}. The upper and lower bounds of power flow variables (i.e., $\overline{P}_{ij}$/$\overline{Q}_{ij}$ and $\underline{P}_{ij}$/$\underline{Q}_{ij}$) are encapsulated by the faulting status of distribution lines. When $\omega_{ij}=0$, the power flows on line $(i,j)$ will be forced as zero. Also, different from the traditional ``Big-M" based expression for capturing the relationship of voltage and line flow status under contingency occasions, a bilinear formulation is developed as in \eqref{RO:CONS9}. It helps avoid the selection of  ``Big-M" factor, and is beneficial for the robustness of model computation \cite{zeng2014chance}. 

\begin{remark}
\label{rmk:1}
Following the ideas in \cite{lei2018routing}, the bilinear structure in \eqref{RO:CONS9} can be readily linearized. For example, the bilinear term $A_{ij}\omega_{ij}=(\tilde{A}_{ij})$, i.e., $A$=$P$, $Q$ or $U$, can be equivalently converted into the following linear constraints:
\setlength{\arraycolsep}{-0.2em}
\begin{eqnarray}
&& \tilde{A}_{ij} \leq \overline{A}_{ij}\omega_{ij}, \quad \forall (i,j) \in \mathcal{B}  \label{RO:CONS11}\\
&& \tilde{A}_{ij} \leq A_{ij}, \quad \forall (i,j) \in \mathcal{B} \label{RO:CONS12}\\
&& \tilde{A}_{ij} \geq A_{ij}-\overline{A}_{ij}(1-\omega_{ij}), \quad \forall (i,j) \in \mathcal{B} \label{RO:CONS13}\\
&& \tilde{A}_{ij} \geq 0, \quad \forall (i,j) \in \mathcal{B} \label{RO:CONS14}
\end{eqnarray}

\end{remark}

In general, the compact matrix expression of $\mathbf{DDU-RSO}$ in \eqref{OBJ}-\eqref{RO:CONS9} is given as below:
\setlength{\arraycolsep}{-0.3em}
\begin{eqnarray}
	&& \mathbf{DDU-RSO}: \ \Gamma=\min_{\bm{x}\in \mathcal{X}} \max_{\bm{u}\in \mathcal{U}(\bm{x})} \sum_{s\in \mathcal{S}} \pi_s \min_{\bm{y}_s \in \mathcal{Q}_s(\bm{u})}  \ \bm{\rho}^T\bm{l^c}_s \label{M:OBJ} \\
	&&s.t. \ \mathcal{X}=\left\lbrace \bm{x}\in \{0,1\}^{{n_1}\times{n_2}}: \ \bm{B} \bm{x} \leq b \right\rbrace \label{M:FS} \\
	&&\mathcal{U}(\bm{x})=\{ \bm{u}\in \{0,1\}^{{m_1}\times{m_2}}: \ \bm{F} \bm{u} \geq \bm{f}, \label{M-DDU1}  \\	
	&&\quad \quad \quad \quad \quad \quad \quad \quad \quad \quad \quad \quad \bm{u} \geq \bm{x} \label{M-DDU2} \} \\
	&&\mathcal{Q}_s(\bm{u})=\{ \bm{y}_s=\left\lbrace \bm{p}_s, \bm{l^c}_s, \bm{v}_s  \right\rbrace \in \mathbb{R}^{e}:  \label{M:CONS0} \\ 
	&&\quad \quad \quad \quad \quad \quad \quad \quad \quad \ \bm{A}\bm{p}_s=\bm{\tilde{l}}_s-\bm{l^c}_s, \label{M:CONS1}  \\
	&&\quad \quad \quad \quad \quad \quad \quad \quad \quad \ \bm{J}\bm{p}_s+\bm{H}\bm{v}_s\geq \bm{C}\bm{u}+\bm{h} \label{M:CONS2} 
\end{eqnarray}

Note that the objective function \eqref{OBJ} is abstracted in \eqref{M:OBJ}. Eq. \eqref{M:FS} represents the hardening budget constraint in \eqref{FS}. The DDU set as a mapping from $\bm{x}$ to $\bm{u}$ is denoted by \eqref{M-DDU1}-\eqref{M-DDU2}. The on-emergency operational status can be defined by three classes of variables, i.e., $\bm{p}_s$ for active/reactive nodal generation (the accumulative outputs of DGs and ESSs) and branch power flows, $\bm{l^c}_s$ for load shedding, and $\bm{v}_s$ for nodal voltage levels. We mention that the power balance requirements (i.e., constraints \eqref{RO:CONS1} and \eqref{RO:CONS2}) are imposed in \eqref{M:CONS1},  where $\bm{\tilde{l}}_s$ is the parameter vector for nominal load demands. Other constraints concerning DGs/ESSs outputs and DistFlow model are denoted by \eqref{M:CONS2}, corresponding to \eqref{RO:CONS3}-\eqref{RO:CONS8} and \eqref{RO:CONS11}-\eqref{RO:CONS14}.

\begin{remark}
\label{rmk:2}
As $\mathcal{U}(\bm{x})$ depends on $\bm{x}$, it  has a complex and variable structure. Some reformulation techniques can be used to convert $\mathcal{U}(\bm{x})$  into a DIU set $\mathcal{U}^0$, so that the basic C\&CG can be applied to compute the associated RO model \cite{Zhang2021TSG}. Recently, a new variant of C\&CG directly exploits the structure of $\mathcal{U}(\bm{x})$~\cite{zeng2022two} and shows a better solution capacity. In the context of distribution systems, we customize it with an enhancement based on  \textit{resilience importance indices}, obtaining a strong method to compute the DDU-RSO problem.
\end{remark}

\section{Enhanced Parametric C\&CG Algorithm}
\label{PCCG}

\subsection{Single-Level Reformulation of DDU-RSO}
To address the computational challenges of $\mathbf{DDU-RSO}$, we first consider its max-min substructure by leveraging the extreme points of recourse constraints set. Due to the slackness of load shedding variables $\bm{l^c}_s$,  we have that the dual problem of \eqref{M:CONS0}-\eqref{M:CONS2}  always has a finite optimal value.  

\setlength{\arraycolsep}{-0.3em}
\begin{eqnarray}
	&&\max_{\bm{u}\in \mathcal{U}(\bm{x})} \sum_{s\in \mathcal{S}} \pi_s \min_{\bm{y}_s \in \mathcal{Q}_s(\bm{u})}  \ \bm{\rho}^T\bm{l^c}_s \label{EQ1}  \\
	&&=\max \{(\bm{C}\bm{u})^T \bm{\theta}+\bm{h}^T \bm{\theta}+\sum_{s\in \mathcal{S}} \bm{\tilde{l}}_s^T\bm{\lambda}_s:\ \bm{u}\in \mathcal{U}(\bm{x}),  \label{EQ2}\\
	&&\bm{J}^T\bm{\vartheta}_s+\bm{A}^T\bm{\lambda}_s=\bm{0}, \bm{\lambda}_s\leq \pi_s \bm{\rho}, \bm{H}^T\bm{\vartheta}_s\leq \bm{0}, \bm{\vartheta}_s\geq \bm{0}, \forall s \}  \label{EQ3}
\end{eqnarray}
where $\bm{\vartheta}=\sum_{s\in \mathcal{S}} \bm{\vartheta}_s$, while $\bm{\lambda}_s$ and $\bm{\vartheta}_s$ denote the vectors of dual variables corresponding to \eqref{M:CONS1} and \eqref{M:CONS2} for each stochastic scenario $s$, respectively. 
\begin{remark}
	\label{rmk:3}
	Let $\Xi_s=\{\bm{J}^T\bm{\vartheta}_s+\bm{A}^T\bm{\lambda}_s=\bm{0}, \bm{\lambda}_s \leq \pi_s \bm{\rho}, \bm{H}^T\bm{\vartheta}_s \leq \bm{0}, \bm{\vartheta}_s \geq 0 \}$. Note that there exists an optimal solution for \eqref{EQ1} with $(\hat{\bm{\lambda}}_s,\hat{\bm{\vartheta}}_s)$ being an extreme point of $\Xi_s$. Unlike $\mathcal{U}(\bm{x})$, $\Xi_s$ is a fixed polyhedron independent of $\bm{u}$ and $\bm{x}$. Consequently, it has a finite and fixed set of extreme points. 
\end{remark}

For fixed $\bm{x}$ and $\bm{\theta}$, the worst-case scenario out of $\mathcal{U}(\bm{x})$ can be characterized as an optimal solution of the following mixed-integer program (MIP):
\setlength{\arraycolsep}{-0.3em}
\begin{eqnarray}
	&&\max \{(\bm{C}^T \bm{\theta})^T \bm{u}:\ \bm{u}\in \mathcal{U}(\bm{x}) \} \label{EQ4}
\end{eqnarray}
Let $\mathcal{U}'(\bm{x})$ be the continuous relaxation of $\mathcal{U}(\bm{x})$. If $\mathcal{U}(\bm{x})$ has the \emph{totally unimodular} property, we have that $\mathcal{U}'(\bm{x})$ equals the convex hull of $\mathcal{U}(\bm{x})$. Actually, due to the special structure of constraints in our $\mathcal{U}(\bm{x})$, this assumption generally holds. Under such a situation, the solution set of \eqref{EQ4} can be exactly defined by its KKT conditions for each $\bm{\theta}_n$:
\begin{eqnarray}
	&&\mathcal{OU}(\bm{x},\bm{\theta}_n)=\left\{ \quad
	\begin{array}{ccc}
		\rm{Eqs.} \ \eqref{M-DDU1}-\eqref{M-DDU2} \\
		\bm{F}^T \bm{\mu}+\bm{\gamma}\geq \bm{C}^T \bm{\theta}_n \\
		\bm{\mu}\circ (\bm{F} \bm{u}_n-\bm{f})=\bm{0}\\
		\bm{\gamma} \circ (\bm{u}_n-\bm{x})=\bm{0}\\
		\bm{u}_n \circ (\bm{F}^T \bm{\mu}+\bm{\gamma}-\bm{C}^T \bm{\theta}_n)=\bm{0}\\
		\bm{0}\leq \bm{u}_n\leq \bm{1}, \bm{\gamma}\leq \bm{0}, \bm{\mu}\leq 0
	\end{array}
	\quad \right\}\label{KKT}
\end{eqnarray}
where $\bm{\mu}$ and $\bm{\gamma}$ denote the dual variables of \eqref{M-DDU1}-\eqref{M-DDU2}. In \eqref{KKT}, Lines 1-2 are constraints that ensure the primal/dual feasibility. Lines 3-5 give the constraints of complementary slackness. Note that $\bm{a}\circ \bm{b}$ in \eqref{KKT} represents a Hadamard product \cite{zeng2022two}, i.e., yielding the element-wise product of vectors $\bm{a}$ and $\bm{b}$.

Let $\Omega_s=\{\bm{\vartheta}_s^1, \ldots, \bm{\vartheta}_s^{N'}\}$ denote the set of extreme points of $\Xi_s$ and let $\Omega$ be the Cartesian product of all $\Xi_s$, i.e., $\Omega= \Omega_1 \times \Omega_2 \times \ldots \times \Omega_{N_S}$. Then, we define $\Theta = \{ \bm{\theta}_1, \ldots, \bm{\theta}_{N}\}$ as the set whose elements are obtained by summing the $N_S$ components in each ordered tuple from $\Omega$.
Following \textbf{Remark \ref{rmk:3}}, by enumerating those KKT-based solution sets, we have the single-level form of $\mathbf{DDU-RSO}$, which yields the basis for the P-C\&CG algorithm.  
\setlength{\arraycolsep}{-0.3em}
\begin{eqnarray}
	&&\mathbf{DDU-RSO}*: \ \Gamma=\min_{\bm{x}\in \mathcal{X}} \ \varphi \label{SL:CONS0}\\
	&&\varphi\geq \sum_{s\in \mathcal{S}} \pi_s \bm{\rho}^T\bm{l}_s^{c,n}, \quad n=1,\ldots,N \label{SL:CONS1}\\
	&&\bm{u}_n\in \mathcal{OU}(\bm{x},\bm{\theta}_n), \quad n=1,\ldots,N \label{SL:CONS2}\\
	&&\bm{A}\bm{p}_s^n+\bm{l}_s^{c,n}=\bm{\tilde{l}}_s, \quad n=1,\ldots,N, \forall s \in \mathcal{S} \label{SL:CONS3}\\ 
	&&\bm{J}\bm{p}_s^n+\bm{H}\bm{v}_s^n\geq \bm{C}\bm{u}_n+\bm{h}, \quad n=1,\ldots,N, \forall s \in \mathcal{S} \label{SL:CONS4}\\ 
	&&\bm{u}_n\in [\bm{0},\bm{1}], \ \bm{p}_s^n, \bm{l}_s^{c,n}, \bm{v}_s^n \in \mathbb{R}^{e},n=1,\ldots,N, \forall s \in \mathcal{S} \label{SL:CONS5}
\end{eqnarray}

\subsection{Algorithm Designs of Parametric C\&CG}
The basic idea of C\&CG algorithm is to construct a master-subproblem computing framework, which iteratively generates new variables and constraints to strengthen the relaxation of the original problem. As shown in section III-A, $\mathbf{DDU-RSO}$ can be reformulated to a single-level form by enumerating the extreme points in $\Theta$. Hence, different from adding cutting planes with fixed critical scenarios in the basic C\&CG method, the P-C\&CG is to dynamically produce new extreme points and the corresponding optimality conditions characterizing non-trivial scenarios. 

Before proceeding to the complete procedures of P-C\&CG, we first provide the master and sub- problems as in the following.

\subsubsection{Subproblem}
For a given first-stage solution $\bm{\hat{x}}$, we define the subproblem ($\mathbf{SP}$) based on \eqref{EQ1}-\eqref{EQ3}.
\begin{eqnarray}
	&&\mathbf{SP}: \ \hat{\varphi}(\bm{x})=\max_{\bm{u}\in \mathcal{U}(\hat{\bm{x}})} \sum_{s\in \mathcal{S}} \pi_s \min_{\bm{y}_s \in \mathcal{Q}_s(\bm{u})}  \ \bm{\rho}^T\bm{l^c}_s	 \label{SP}	   	
\end{eqnarray}

The objective is to minimize the load shedding under the worst-case scenario. The bilevel problem $\mathbf{SP}$ can be computed by transforming the inner-level problem into its dual form, which is given as \eqref{EQ2}-\eqref{EQ3}. Due to the load shedding variables $\bm{l^c}_s$ in the power balance equations, $\mathbf{SP}$ holds the relatively complete resource property. It guarantees that the optimal solution of $\mathbf{SP}$ exists for any given $\hat{\bm{x}}$, which corresponds to an extreme point, i.e., $(\hat{\bm{\lambda}},\hat{\bm{\theta}})$. The dual solution would yield a set of optimality condition cuts as in \eqref{KKT}.

\subsubsection{Master Problem}
Then, the optimality cutting set will be created and added to the master problem ($\mathbf{MP}$).
\begin{eqnarray}
	&&\mathbf{MP}_l: \ 
	\Gamma=\min_{\bm{x}\in \mathcal{X}} \ \varphi \label{MP:CONS0} \\
	&&\varphi\geq \sum_{s\in \mathcal{S}} \pi_s \bm{\rho}^T\bm{l}_s^{c,n}, \quad n=1,\ldots,l-1 \label{MP:CONS1}\\
	&&\bm{u}_n\in \mathcal{OU}(\bm{x},\hat{\bm{\theta}}_n), \quad n=1,\ldots,l-1 \label{MP:CONS2}\\
	&&\bm{A}\bm{p}_s^n+\bm{l}_s^{c,n}=\bm{\tilde{l}}_s, \quad n=1,\ldots,l-1, \forall s \in \mathcal{S} \label{MP:CONS3}\\ 
	&&\bm{J}\bm{p}_s^n+\bm{H}\bm{v}_s^n\geq \bm{C}\bm{u}_n+\bm{h}, \quad n=1,...,l-1, \forall s \in \mathcal{S} \label{MP:CONS4}\\ 
	&&\bm{u}_n\in [\bm{0},\bm{1}], \ \bm{p}_s^n, \bm{l}_s^{c,n}, \bm{v}_s^n \in \mathbb{R}^{e},n=1,...,l-1,, \forall s \in \mathcal{S} \label{MP:CONS5}
\end{eqnarray}
where $l$ represents the index of iteration. Note that $\mathbf{MP}_l$ is a relaxation of $\mathbf{DDU-RSO}*$. By expanding a partial enumeration of $\Theta$ and corresponding non-trivial scenario set gradually, the lower bound can be improved iteratively till the convergence condition is satisfied.

\subsubsection{Algorithm Procedures}
Based on well-defined $\mathbf{MP}$ and $\mathbf{SP}$, the customized algorithm procedures are given as below.

\noindent $\blacksquare$ \textbf{Algorithm 1: P-C\&CG}
\begin{description}
	\item [$\rm{Step \ 1}$] $\ $ Initialize $UB=\infty$, $LB=-\infty$, $l=1$ and ${\Theta}'=\emptyset$.
	\item [$\rm{Step \ 2}$] $\ $ Solve the master problem $\mathbf{MP}_l$; Derive the solution $\hat{\bm{x}}^l$ and its optimal value $\hat{\Gamma}$. Update $LB = \hat{\Gamma}$.
	\item [$\rm{Step \ 3}$] $\ $ For $\hat{\bm{x}}^{l}$, solve the subproblem $\mathbf{SP}_l$,
	to attain the critical contingency scenario $\hat{\bm{u}}_{l}$, corresponding extreme point $(\hat{\bm{\lambda}}_l,\hat{\bm{\theta}}_l)$ and optimal value $\hat{\varphi}(\hat{\bm{x}}^{l})$.
	\item [$\rm{Step \ 4}$] $\ $ Update ${\Theta}'={\Theta}'\cup \{\hat{\bm{\theta}}_l\}$, and strengthen $\mathbf{MP}_l$ by creating and adding the following cutting set:
	\setlength{\arraycolsep}{-0.4em}
	\begin{eqnarray}
		&&\varphi\geq \sum_{s\in \mathcal{S}} \pi_s \bm{\rho}^T\bm{l}_s^{c,l}  \label{Cut1}  \\ 
		&&\bm{u}_l\in \mathcal{OU}(\bm{x},\hat{\bm{\theta}}_l) \label{Cut2}  \\ 
		&&\bm{A}\bm{p}_s^l+\bm{l}_s^{c,l}=\bm{\tilde{l}}_s, \quad \forall s \in \cal S  \label{Cut3}  \\ 
		&&\bm{J}\bm{p}_s^l+\bm{H}\bm{v}_s^l\geq \bm{C}\bm{u}_l+\bm{h}, \quad \forall s \in \cal S \label{Cut4}\\ 
		&&\bm{u}_l\in [\bm{0},\bm{1}], \ \bm{p}_s^l, \bm{l}_s^{c,l}, \bm{v}_s^l \in \mathbb{R}^{e},  \forall s \in \cal S \label{Cut5} 
	\end{eqnarray}
	\item [$\rm{Step \ 5}$] $\ $ Update $UB$ = min$\left\lbrace UB,\hat{\varphi}(\hat{\bm{x}}^{l})\right\rbrace$ .
	\item [$\rm{Step \ 6}$] $\ $ If $\frac{UB-LB}{LB}\leq \varepsilon$, terminate and report the hardening solution. Otherwise, set $l\gets {l+1}$ and go ot $\rm{Step \ 2}$.
\end{description}

\begin{remark}
\label{rmk:4}
Notice again that $\Theta$ is a fixed and finite set, so the convergence property of P-C\&CG can be easily deduced. Moreover, if an optimal solution $x^*$ to $\mathbf{MP}$ exists in two iterations, the algorithm converges in the latter iteration. Hence, together with the fact that the number of candidate hardening schemes is finite, the cardinality of $\mathcal{X}$ would further bound the number of iterations. 
\end{remark}

\subsection{Comparative Analysis with the Basic C\&CG}
As mentioned in \textbf{Remark \ref{rmk:1}}, $\mathbf{DDU-RSO}$ can be reformulated as a DIU-based model. Specifically, different from the DDU set in \eqref{DDU:CONS1}-\eqref{DDU:CONS4}, the DIU set $\mathcal{U}^0$ only contains the $N-k$ constraint.  Then, we can make use of $\bm{u}'=\bm{u}+\bm{x}-\bm{u}^*\circ\bm{x}$ to modify the contingency status in the second-stage recourse problem, leading to an equivalent DIU-based formulation. It can be solved by the basic C\&CG method, through adding the following cutting sets iteratively: 
\begin{eqnarray}
	&&\varphi\geq \sum_{s\in \mathcal{S}} \pi_s \bm{\rho}^T\bm{l}_s^{c,l}  \label{DIU-Cut1}  \\ 
	&&\bm{A}\bm{p}_s^l+\bm{l}_s^{c,l}=\bm{\tilde{l}}_s, \quad \forall s \in \cal S   \label{DIU-Cut2}  \\ 
	&&\bm{J}\bm{p}_s^l+\bm{H}\bm{v}_s^l\geq \bm{C}\bm{u}_l+\bm{h}, \quad \forall s \in \cal S \label{DIU-Cut3}\\ 
	&&\bm{u}'_l=\bm{u}_l^*+\bm{x}-\bm{u}^*_l \circ \bm{x} \label{DIU-Cut4}\\ 
	&&\bm{u}'_l\in \{\bm{0},\bm{1}\}, \ \bm{p}_s^l, \bm{l}_s^{c,l}, \bm{v}_s^l \in \mathbb{R}^{e},  \forall s \in \cal S \label{DIU-Cut5} 
\end{eqnarray}
where $\bm{u}^*_l$ denotes the worst-case line status of \textbf{SP} in iteration $l$. Between the P-C\&CG and the basic C\&CG, the following results highlight the advantage of the former one. 
\begin{lem}
\label{lem:1}
For the inner-most problem of $\mathbf{DDU-RSO}$, the more lines and DGs are damaged, the larger the optimal value is. 
\end{lem}

\begin{proof}
The inner-most problem of $\mathbf{DDU-RSO}$ is a linear program, whose objective function is to minimize the total load shedding of the system. Line $(i,j)$ being damaged means $u_{ij}=0$, causing the upper and lower bound of the constraints \eqref{RO:CONS7} and \eqref{RO:CONS8} to be zero. Similarly, when the DGs deployed at node $j$ are destroyed, their available generation capacity becomes zero, as shown in \eqref{RO:CONS5} and \eqref{RO:CONS6}.
Hence, the greater the number of damaged lines and DGs, the smaller the feasible region is, and hence the larger the optimal objective value is. 
\end{proof}

\begin{prop}
\label{prop:1}
Consider iteration $l$ and denote the given $\mathbf{MP}$ by $\mathbf{MP}_l$. After $\mathbf{SP}$ is solved, the parametric cutting set in the form of \eqref{Cut1}-\eqref{Cut5} or the basic one in the form of \eqref{DIU-Cut1}-\eqref{DIU-Cut5} can be appended to $\mathbf{MP}_l$ to obtained $\mathbf{MP}_{l+1}$. $\mathbf{MP}_{l+1}$ with the parametric cutting set is stronger, i.e., the optimal value  of $\mathbf{MP}_{l+1}$ with  the parametric one is larger than or equal to that of $\mathbf{MP}_{l+1}$ obtained by augmenting $\mathbf{MP}_l$ with the basic one.  
\end{prop}                                                 

\begin{proof}

Note that \eqref{EQ4} can be decomposed into two independent parts, i.e., $\mathcal{OU}(\bm{z},\bm{\theta^\bm{\omega}})$ and $\mathcal{OU}(\bm{r},\bm{\theta^\bm{\nu}})$ with $\bm x=(\bm z,\bm r)$ and $\bm \theta=(\bm \theta^\omega,\bm\theta^{\nu})$. For the sake of illustration, we first analyze the properties of $\mathcal{OU}(\bm{z},\bm{\theta^\bm{\omega}})$.	
	
According to the duality theory, we have $\bm{C}^T \bm{\theta^\bm{\omega}}=-(\overline{\bm{P}} \circ (\bm{\pi}^{p,u}+\bm{\pi}^{p,l})+\overline{\bm{Q}}\circ (\bm{\pi}^{q,u}+\bm{\pi}^{q,l}))$, where $\bm{\pi}^{p,u}$, $\bm{\pi}^{p,l}$, $\bm{\pi}^{q,u}$, and $\bm{\pi}^{q,l}$ are non-negative dual variables of \eqref{RO:CONS7} and \eqref{RO:CONS8}, respectively. When $\mathbf{SP}$ derives its optimal $\bm{\omega}^*_l$, constraint \eqref{DDU:CONS1} is satisfied with equality, meaning that $k^L$ components of vector $\bm{\omega}^*_l$ are equal to 0 while the others are set to 1. 

Now, we can define $\tilde{\bm{\theta}}^{\bm{\omega}}_l$ such that $(\bm{C}^T \tilde{\bm{\theta}}^{\bm{\omega}}_l)^T$ equals $-M$ with \textit{M} being a sufficiently positive value whenever $\bm{\omega}^*_{ij}=0$. As noted in \cite{zeng2022two}, this modification ensures that $\mathcal{OU}(\bm{z},\tilde{\bm{\theta}}^{\bm{\omega}}_l)$ is a singleton that only contains  $\bm{\omega}^*_l$, as long as the set of lines being hardened  does not overlap the lines corresponding to zero entries of $\bm{\omega}^*_l$. 

We consider the following two cases of $\textbf{MP}_{l+1}$ where it is with cutting set \eqref{DIU-Cut1}-\eqref{DIU-Cut5} and \eqref{Cut1}-\eqref{Cut5}, respectively. 
\textbf{(Case I)} If the hardened lines of $\bm z$ and zero entries of $\bm \omega^*_l$ do not overlap, it can be seen that $\mathcal{OU}(\bm{z},\tilde{\bm{\theta}}^{\bm{\omega}}_l)=\{\bm \omega^*_l\}$, and the parametric cutting set, i.e., \eqref{Cut1}-\eqref{Cut5}, defined with respect to $\mathcal{OU}(\bm{z},\tilde{\bm{\theta}}^{\bm{\omega}}_l)$  and \eqref{DIU-Cut1}-\eqref{DIU-Cut5} defined on $\bm{\omega}^*_l$ are identical. 

(Case II) If $k'$ of the $k^L$ damaged lines in $\bm{\omega}^*_l$ are hardened, let $\hat{\bm z}'$ be the set of those lines. According to \eqref{DIU-Cut4}, only the remaining lines in $\bm{\omega}^*_l\backslash \hat{\bm z}'$ are damaged. However, given the structure of linear constraints \eqref{DDU:CONS1} and \eqref{DDU:CONS2} and the linear objective function $(\bm{C}^T \tilde{\bm{\theta}}^{\bm{\omega}}_l)^T$, the worst-case scenario derived from the $\mathcal{OU}(\bm{z},\tilde{\bm{\theta}}^{\bm{\omega}}_l)$ contains two parts. The first part is those in $\bm{\omega}^*_l\backslash \hat{\bm z}'$. The second part contains additional $k'$ lines that are not hardened.

Given that  $\bm{\omega}^*_l\backslash \hat{\bm z}'$ is a subset of $\mathcal{OU}(\bm{z},\tilde{\bm{\theta}}^{\bm{\omega}}_l)$,   according to \textbf{Lemma 1}, it can be seen that the recourse problem's optimal value for the parametric cutting set \eqref{Cut1}-\eqref{Cut5} is larger than or equal to that of \eqref{DIU-Cut1}-\eqref{DIU-Cut5}. 

Note that the structure and property of \eqref{DDU:CONS3}-\eqref{DDU:CONS4} are the same as \eqref{DDU:CONS1}-\eqref{DDU:CONS2}.  So, the same conclusion can be made for $\mathcal{OU}(\bm{r},\tilde{\bm{\theta}}^{\bm{\nu}}_l)$ by following a similar approach.
  
Consequently, the optimal value of $\mathbf{MP}_{l+1}$ in the P-C\&CG framework is larger than or equal to that of $\mathbf{MP}_{l+1}$ obtained with the basic cutting set.
\end{proof}

\subsection{Enhancement Strategy Based on Structural Information}
According to the proof for \textbf{Proposition 1}, we have noted that \eqref{M-DDU1}-\eqref{M-DDU2} often has multiple optimal solutions, rendering $\mathcal{OU}$ a non-singleton set. With this issue, the performance of P-C\&CG is less consistent, since extra iterations could be involved with a longer computational time. To handle this challenge, we introduce the concept of \textit{resilience importance indices} by fully leveraging the structural property of PDN. 
For the sake of illustration, the resilience importance index of line $(i,j)$ is denoted by $\xi_{ij}$. It captures the impact of a single outage of line $(i,j)$ on network resilience. Since our formulation cares about the load shedding, $\xi_{ij}$ is evaluated as the load shedding losses after disconnecting every line $(i,j)$. Specifically, given a damage scenario $\bm{\omega}$, where $\omega_{ij}=0$ and other elements fixed to 1, the value of $\xi_{ij}$ can be derived by solving the lower-level problem $\sum_{s\in \mathcal{S}} \pi_s \min_{\bm{y}_s \in \mathcal{Q}_s(\bm{u})} \bm{\rho}^T\bm{l^c}_s$ under this scenario.
Clearly, such an index contains the prior knowledge on the importance of different power distribution lines.

Let $\bm{\xi}$ be the vector of resilience importance indices. We next slightly modify \eqref{EQ4} using $\bm{\xi}$.   
\begin{eqnarray}
	&&\max \{(\bm{C}\bm{u})^T \bm{\theta} + \epsilon \cdot \bm{\xi}^T \bm{u}:\ \bm{u}\in \mathcal{U}(\bm{x}) \} \label{EQ5}
\end{eqnarray}
where $\epsilon$ is a very small negative number. For the same hardening decision $\bm{x}$, under the adjustment of resilience importance indices $\bm{\xi}$, the optimal solution reduces to a singleton or a much small subset of the original optimal solution set, which could result in larger load shedding due to the outage. Consequently, the utilization of resilience importance indices provides an effective way to identify the lines that are more adversarial to hardening, and thus significantly improve the computational efficiency compared to the standard  P-C\&CG.
Note that our intention to single out an optimal solution subset that is more likely to be against hardening, rather than directing \eqref{EQ4} to totally different optimal solutions.

\begin{remark}
\label{rmk:5}	
We would like to mention that our utilization of resilience importance indices produces an effect similar to the well-known Pareto-optimality adopted in classical Benders decomposition algorithm, i.e., identifying an optimal solution (or optimal solution subset) that more adversarial to hardening. Note that it is also achieved by the deep knowledge on the underlying power distribution system. By employing the corresponding $\mathcal{OU}$ set, we can effectively deal with the multi-solution issues and thus significantly improve the computational efficiency.
\end{remark}

\section{Numerical Results}
\label{Numerical}
The effectiveness and computational feasibility of DDU-RSO based resilient hardening method are validated on 33-bus and 118-bus test distribution systems. We consider the system hardening decisions against a hurricane and its cascading effect of regional flooding. A number of overhead power lines are buried underground to prevent the strike of strong winds and rain storms. Also, the installation position of DGs is evaluated to protect them from floodwaters. For simplicity, the budget constraint in \eqref{FS} is replaced by the upper restriction on total number of hardened lines and DGs, which is represented by a parameter $\overline{\Upsilon}$. Additionally, the typical scenarios of power loads and initial storage levels are generated using the stratified sampling techniques (for more details can be seen in our previous studies \cite{sun2022multistage,cao2021risk,cao2021hydrogen}). 

Our tests are performed in MATLAB R2024b on a desktop computer with Intel Core i9-10900K 3.70GHz processor and 128GB memory. All the $\mathbf{MP}$s and $\mathbf{SP}$s are computed using Gurobi 11.0.

\subsection{33-Bus Test Distribution Network}
Fig. \ref{fig:PND} presents a modified structure of IEEE 33-bus distribution network, where 5 DGs and 3 ESS units are configured with identical power rating as 500kW. The nominal storage capacity of each ESS unit is set to 1500 kWh, with a random fluctuation range of $\pm$20\%. The DGs installed at N4, N11, N14, N18 and N33 are indexed from DG1 to DG5.  The resilient hardening decisions on 33-bus system are exhibited and analyzed as in the following.  

\begin{figure}[!t]
	\centering
	\includegraphics[width=3.5in]{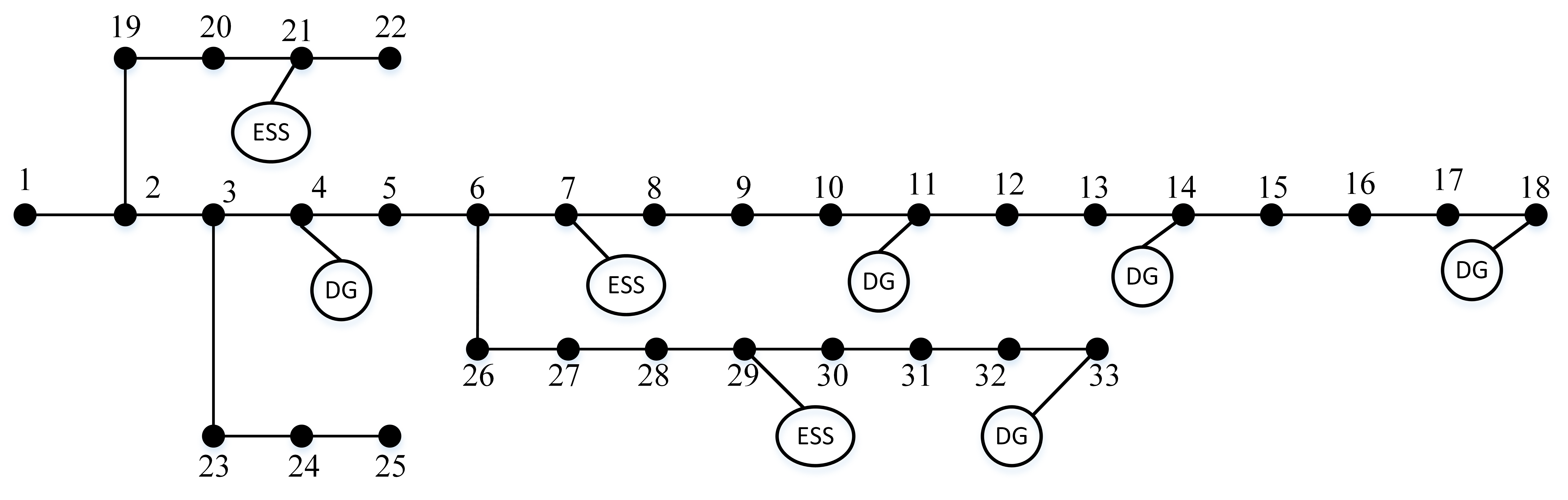}
	\caption{Demonstration of 33-bus Power Distribution Network} 
	\label{fig:PND} 
\end{figure}

\subsubsection{Resilient Hardening Results Under Different Budget Levels} 
First, the $\mathbf{DDU-RSO}$ model is computed to obtain the network hardening schemes under different budget levels. We set up four cases where  $\overline{\Upsilon}=$ 0, 2, 4, 6. In our tests, the $k^L$ and $k^{DG}$ are considered as 5 and 1, indicating no more than 5 distribution lines and 1 DG will be damaged under catastrophic occasions. The hardening results with associating worst-case contingency scenarios are presented in Fig. \ref{fig:HS}. 

\begin{figure*}[!t]
	\centering
	\includegraphics[width=5.0in]{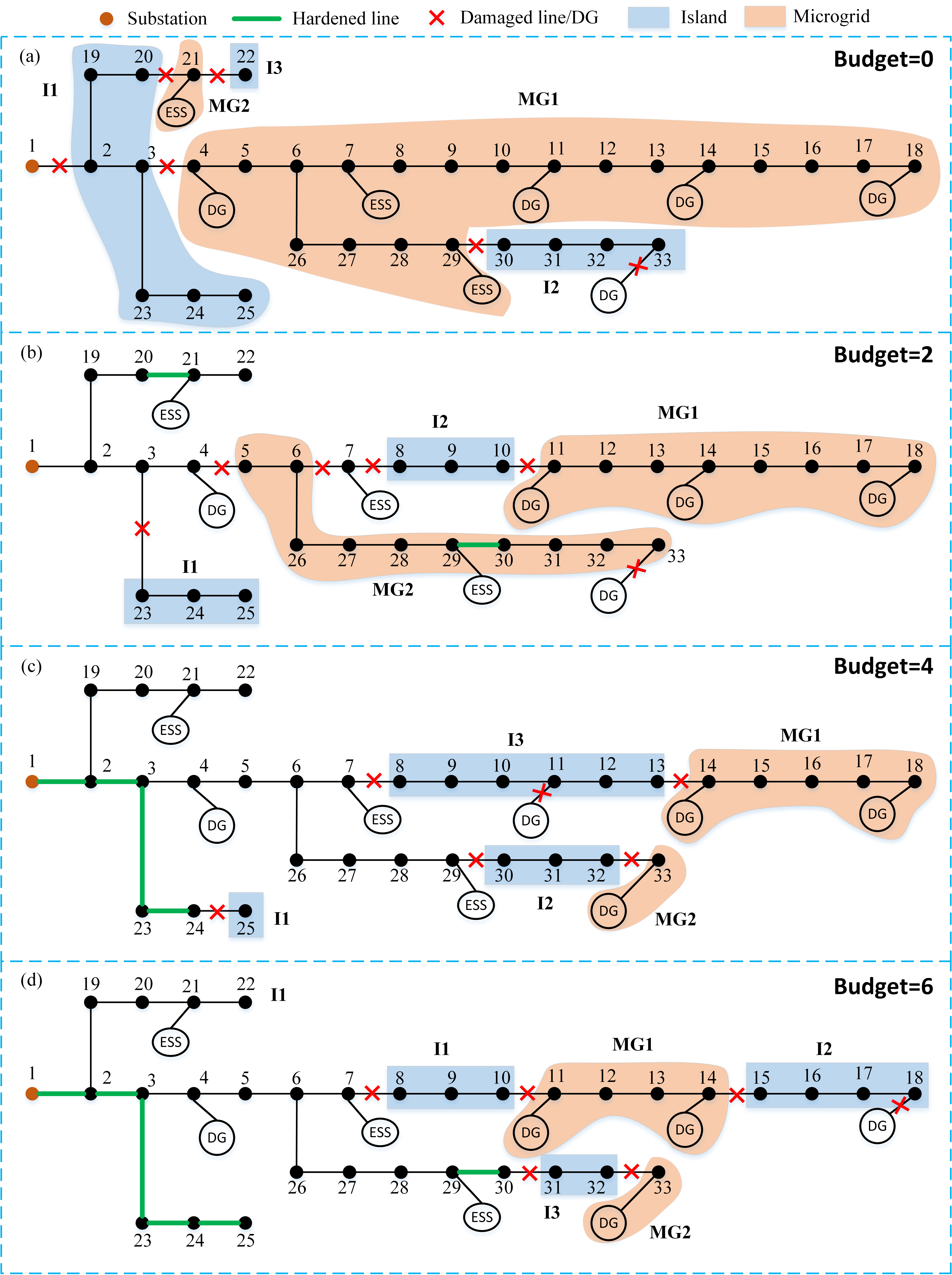}
	\caption{Resilient Hardening Results under Different Budget Levels ($\overline{\Upsilon}$): 33-Bus Test System} 
	\label{fig:HS} 
\end{figure*}

\textit{$\bullet$ Case 0} ($\overline{\Upsilon}=0$):
With no hardening actions, the $\mathbf{DDU-RSO}$ model in \eqref{M:OBJ}-\eqref{M:CONS2} reduces to a DIU-based RSO problem. In this baseline case, the most serious contingency occurs when the attack of wind storm cuts off the distribution lines L1-2, L3-4, L20-21, L21-22, L29-30, while the cascading flood forces DG5 to be shut down. As shown in Fig. \ref{fig:HS}(a), the entire system is unintentionally partitioned into 5 parts, including three islands without any power sources, as well as two dynamic microgrids that preserve the critical loads supply using DGs and ESSs. The catastrophic consequence is the load shedding of more than 54\%.

\textit{$\bullet$ Case 1} ($\overline{\Upsilon}=2$):
In this case, the distribution lines L20-21 and L29-30 are hardened. So they would not be damaged as in \emph{Case 0}. Alternatively, the worst-case contingency scenario is identified as the outages of L4-5, L6-7, L7-8, L10-11, L3-23 and DG5. As a result, the islanding areas shrink from 12 electric nodes to 6, which thus decreases the ratio of load shedding from 54.1\% to 48.6\%. According to Fig. \ref{fig:HS}(b), the protection of L20-21 seems rather effective. Together with the integration of ESS unit at N21, the damage on either of the line sections among L1-2-19-20 would not incur islanding.

\textit{$\bullet$ Case 2} ($\overline{\Upsilon}=4$):
When the hardening budget covers four distribution lines, the protection of L21-21 and L29-30 are strategically abandoned. Instead, the lines L1-2, L2-3, L3-23 and L23-24 are hardened. As presented in Fig. \ref{fig:HS}(c), the worst-case damage scheme is significantly changed, i.e., the outages on L13-14, L24-25, L29-30, L32-33 and DG2 replace those on L4-5, L6-7, L10-11, L3-23 and DG5. 
The rationality of protecting L3-23 and L23-23 can be explained as the high power demands at nodes N23 and N24, which is more valuable than the reinforcement of L29-30. Also, more power demands can be supplied via the distribution substation (N1), while the self-energized regions (i.e., the microgrids) are largely reduced.  
As a result, and the ratio of load shedding further reduces to 38.3\%.

\textit{$\bullet$ Case 3} ($\overline{\Upsilon}=6$):
When $\overline{\Upsilon}$ increases from 4 to 6, two additional distribution lines L24-25 and L29-30 are hardened, on the basis of \emph{Case 2}. As shown in Fig. \ref{fig:HS}(d), the worst-case contingency scenario is identified as the outages of  L7-8, L10-11, L14-15, L30-31, L32-33 and DG4. In contrast to \emph{Case 2}, the further protection of L29-30 preserves the power supply of N30 under the worst-case scenario. On the other hand, the protection of L24-25 forces the change of damage distribution between N10 and N18, which may cause less amount of load losses. The resulting load shedding proportion decreases from 38.3\% to 25.8\%. Compared to the baseline case, we observe a significant improvement in system resilience, as the load shedding is reduced by more than a half.   

The aforementioned test results validate the effectiveness of the proposed hardening method. Our $\mathbf{DDU-RSO}$ model can capture the practical correlation of proactive hardening decisions with the endogenous uncertainty of system contingencies. With the increasing strength of hardening actions, the catastrophic results brought by disastrous events have been gradually relieved.

\subsubsection{Impact of System Damage Levels}
In this section, we fix the budget level at $\overline{\Upsilon}=4$ and DG damage level at $k^{DG}=1$, and then investigate the optimal hardening results under different levels of lines damage severity (represented by parameter $k^L$). The test results are shown in Table \ref{tab:testB}. To make the exposition more explicit, the objective value is expressed in a proportional form (i.e., the ratio of load shedding). When $\overline{\Upsilon}=4$ and $k^L=2$, the distribution lines L3-4 and L3-23 are hardened to protect the N23, which has high power demand. DG4 and DG5 are also protected to avoid losing the energy resources. The worst-case scenario is identified as the faulting of lines L23-24 and L29-30, which leads to the power interruption of N24 and N25, and insufficient power supply of N30-N31-N32-N33 area. When $k^L$ increases to 4, the distribution lines L1-2, L2-3, L3-23 and L23-24 are hardened to ensure the power supply through upstream grid and protect the area with high power demands. The damages are imposed on L7-8, L10-11, L24-25, L29-30 and DG5, which cuts off the power supply of N8-N10, N25, and N30-N33. As a result, the ratio of load shedding increases from 23.7\% to 35.4\%. Then, under $N-6$ contingencies, the hardening scheme remains while additional lines L14-15 and L32-33 are disrupted. Besides, DG5 is expected to be damaged instead of DG4, resulting in a new island, i.e., I2 in Fig. \ref{fig:HS}(d). The load shedding proportion further raises to 41.3\%. Finally, when $k^L=8$, we have L11-14, L17-18 and L21-22 as additionally damaged lines, while the faulting generator alters to DG3. The islanding area expands to N12-N14 and N22, which results in the growth of load shedding proportion to 47.4\%. 

\renewcommand\arraystretch{1.2}
\begin{table*}[!t]
	\centering
	\caption{Hardening Results under Different Damage Levels: 33-Bus Test System}
	\setlength{\tabcolsep}{2.0mm}
	\scalebox{1.0}{
		
		\begin{tabular}{ccccc}
			\toprule
			$k^L$  & $k^{DG}$ &  Hardening Scheme & Worst-Case Contingency Scenario & OBJ  \\
			\midrule		
			2 & 1&  L3-4, 3-23; DG4, DG5  & L23-24, 29-30; DG2 &  23.7\% \\
			4 & 1& L1-2, 2-3, 3-23, 23-24  & L7-8, 10-11, 24-25, 29-30; DG5 & 35.4\%  \\	
			6 & 1& L1-2, 2-3, 3-23, 23-24  & L7-8, 10-11, 14-15, 24-25, 29-30, 32-33; DG4 & 41.3\%\\
			8 & 1& L1-2, 2-3, 3-23, 23-24  & L7-8, 10-11, 11-12, 17-18, 21-22, 24-25, 29-30, 32-33; DG3  & 47.4\% \\	
			\bottomrule	  
	   \end{tabular}}
	\label{tab:testB}
\end{table*}

Clearly, with the growing severity of network contingencies, the proportion of load shedding increases monotonously. However, an interesting observation is the unchanged hardening plan for the test instances where $k^L \geq 4$, which could be attributed to the following reasons: (i) The power import through substation node is very important for ensuring the resilient operation, so that L1-2 and L2-3 should always be protected. (ii) The power demands at N23 and N24 are higher than the rest electric nodes, leading to the reinforcement of L3-23 and L23-24 in priority.

\subsubsection{Comparative Solution Tests of P-C\&CG with Other Methods}
To verify the computational feasibility of P-C\&CG algorithm, it is tested in comparison with two benchmark methods, i.e., the basic C\&CG \cite{zeng2013solving} and modified Benders decomposition (BD) \cite{Zhang2022}. By applying the basic C\&CG, we first reformulate $\mathbf{DDU-RSO}$ into a DIU-based model by using the method in section III-C. Different from C\&CG, the modified BD algorithm adopts the dual information of recourse problems to construct the cutting planes. Also, it depends on an extra subproblem to guarantee the robust feasibility. 

Table \ref{tab:OS} presents the solution performance of different algorithms under varying levels of line damage severity ($k^L$) and hardening budget ($\overline{\Upsilon}$). The maximum number of damaged DGs ($k^{DG}$) is fixed at 1. The objective value is expressed in a proportional form (i.e., the ratio of load shedding).
Also, the number of required iterations (N\_itr) and computation time (counted in seconds) are given as performance metrics. The test instances that fail to converge in 12 hours are marked by ``T''. It can be seen that all these solution methods can exactly obtain the global optimum. But the basic C\&CG and P-C\&CG demonstrate a superior scalable capability than the modified BD, since the last one could be more sensitive to the change of model parameters. With small values of $k^L$ and $\overline{\Upsilon}$ (e.g., $k^L=4$, $\overline{\Upsilon}=4$), which generally correspond to a small-size uncertainty set, the modified BD spends less time to compute the $\mathbf{DDU-RSO}$ than C\&CGs. However, by increasing either $k^L$ or $\overline{\Upsilon}$, it may lead to an exponential increase of both the iteration number and computation time. In contrast, the solution performance of C\&CGs keeps steady, or even gets better sometimes. For example, given $k^L=8$ and increasing  $\overline{\Upsilon}$ from 4 to 6, the modified BD needs 2.6 times more iterations (increasing from 108 to 396) to converge, which leads to an increase of computation time from 422.9s to 6465.2s. In such instance, P-C\&CG only needs 22 iterations and 210.2s (even less than the test case with $k^L=8$ and $\overline{\Upsilon}=4$) to attain the optimal solution. The computation time is significantly reduced by more than 97\%. When $k^L=8$ and $\overline{\Upsilon}=8$, the P-C\&CG takes 1358.3s (nearly 22.6 minutes) to exactly solve the $\mathbf{DDU-RSO}$, while BD cannot converge within 12 hours. So, for those challenging test cases with huge uncertainty sets, the C\&CG algorithms can be orders-of-magnitudes faster than BD. 

\begin{remark}
\label{rmk:6}
Note that with the growing of iteration number, the $\mathbf{MP}$s could become intractable due to the exponential increase of problem complexity. As indicated by our test results and those in \cite{zeng2022two}, the parametric representation of primal cuts as in \eqref{Cut1}-\eqref{Cut5} is with a higher strength than the Benders-type dual cutting planes. So different from BD, the P-C\&CG often has a good and steady convergence performance with the change of model parameters. It explains the time reduction achieved by P-C\&CG especially for those large-scale instances. 
\end{remark}

Additionally, the computation of $\mathbf{DDU-RSO}$ formulation using P-C\&CG is generally faster than solving its DIU-based equivalence. Although the computing of $\mathbf{MP}$s in P-C\&CG could be more time-consuming due to the complex structure of $\mathcal{OU}$ set, it usually requires fewer iterations to converge by including more abundant recourse information. For example, when $k^L=8$ and $\overline{\Upsilon}=8$, it needs 49 iterations and 1472.1s to achieve the global optimality of DIU-based RSO problem. By directly computing $\mathbf{DDU-RSO}$, the iteration number is reduced to 42, which saves the solution time by 8\%. These results are consistent with the conclusion of \cite{zeng2022two}, i.e., the DDU modeling is actually beneficial for designing more efficient algorithms than its DIU-based equivalence.

\renewcommand\arraystretch{1.0}
\begin{table}[!t]
	\centering
	\caption{Computational Performance of Different Solution Methods: 33-Bus Test System}
	\setlength{\tabcolsep}{1.5mm}
	\scalebox{1}{	
		\begin{tabular}{ccccccccc}
			\toprule
			\multirow{2}{*}{$k^L$} & \multirow{2}{*}{$\overline{\Upsilon}$}  & \multirow{2}{*}{OBJ} & \multicolumn{2}{c}{P-C\&CG} & \multicolumn{2}{c}{Basic C\&CG} & \multicolumn{2}{c}{Modified BD}  \\
			&    &   & N\_itr & Time/s   & N\_itr & Time/s  & N\_itr & Time/s   \\
			\midrule
			\multirow{3}{*}{4} & 4 & 35.4\%  & 15     & 69.0   & 16    & 75.8     & 23    & 37.6    \\
			& 6    & 22.9\%   & 13   & 56.8     & 14       & 68.8      & 50      & 115.3         \\
			& 8    & 19.8\%     &20   &140.3     & 22      & 160.5    & 93    & 331.5     \\
			\midrule
			\multirow{3}{*}{6} & 4   & 41.3\%   &17   &105.5     & 18     & 112.2     & 44      & 125.9       \\
			& 6    & 29.8\%  &18   &119.5        & 19      & 125.9      & 127      & 526.5       \\
			& 8   & 25.0\%  &33   &526.3     & 35      & 469.9     & 323      & 3864.7
			    \\
			\midrule
			\multirow{3}{*}{8} & 4  & 47.4\%   & 25   &226.4     & 26      & 248.9      & 108      & 422.9       \\
			& 6  &  35.0\%  & 22   &210.2      & 24      & 214.6      & 396     & 7165.2           \\
			& 8  & 29.9\%   &42  &1358.3       & 49      & 1472.1     & /      & T        \\
			\bottomrule
	\end{tabular}}
	\label{tab:OS}
\end{table}

\subsubsection{Solution Tests for Enhanced P-C\&CG}
We further investigate the enhancement strategy for P-C\&CG, and present the results of solution tests in Table \ref{tab:PP}. Our enhanced customization is effective, which, compared to the standard procedures of P-C\&CG, has reduced the computation time by exceeding 50\% for most test instances. For example, when $k^L=8$ and $\overline{\Upsilon}=6$, the P-C\&CG originally takes 22 iterations and 210.2s to compute $\mathbf{DDU-RSO}$. After applying our enhancement strategy, it only needs 7 iterations to converge. The resulting computation time is 23.8s, which is decreased by 88.7\%. The iterative procedures of both methods are illustrated in Fig. \ref{fig:IP}. We notice an evident enhancement in convergence performance, where both the $UB$ and $LB$ change more sharply. So the incorporating of resilience importance indices helps yield tightened bounds, leading to a faster convergence without hurting the solution exactness. 

Moreover, by comparing the enhanced P-C\&CG with modified BD, a more significant improvement in solution performance  can be observed. Even for those ``small-scale'' cases (e.g., $k^L=4$, $\overline{\Upsilon}=4$), the enhanced P-C\&CG spends less computation time. Except for this case, the time saving achieved by our enhanced algorithm ranges from 64.3\% to over 99.7\%.

\renewcommand\arraystretch{1.0}
\begin{table}[!t]
	\centering
	\caption{Computational Performance of Enhanced P-C\&CG vs. Its Original Version: 33-Bus Test System}
	\setlength{\tabcolsep}{4mm}
	\scalebox{1}{
		\begin{tabular}{cccccc}
			\toprule
			\multirow{2}{*}{$k^L$} & \multirow{2}{*}{$\overline{\Upsilon}$}  & \multicolumn{2}{c}{P-C\&CG} & \multicolumn{2}{c}{Enhanced P-C\&CG}   \\
			&       & N\_itr & Time/s   & N\_itr & Time/s   \\
			\midrule
			\multirow{3}{*}{4} & 4   & 15    & 69.0    & 10   & 28.1     \\
			& 6   & 13   & 56.8      & 10   & 41.2          \\
			& 8   & 20  &140.3     &10   &39.6       \\
			\midrule
			\multirow{3}{*}{6} & 4     &17   &105.5    &9   &33.6     \\
			& 6    &18   &119.5       &11   &58.6   \\
			& 8    &33   &526.3     &14   &78.1    \\
			\midrule
			\multirow{3}{*}{8} & 4    & 25   &226.4     & 9   &36.9       \\
			& 6     &22   &210.2     &7   &23.8      \\
			& 8     &42   &1358.3    &19  &189.1  \\
			\bottomrule
	\end{tabular}}
	\label{tab:PP}
\end{table}

\begin{figure}[!]
	\centering
	\includegraphics[width=3.3in]{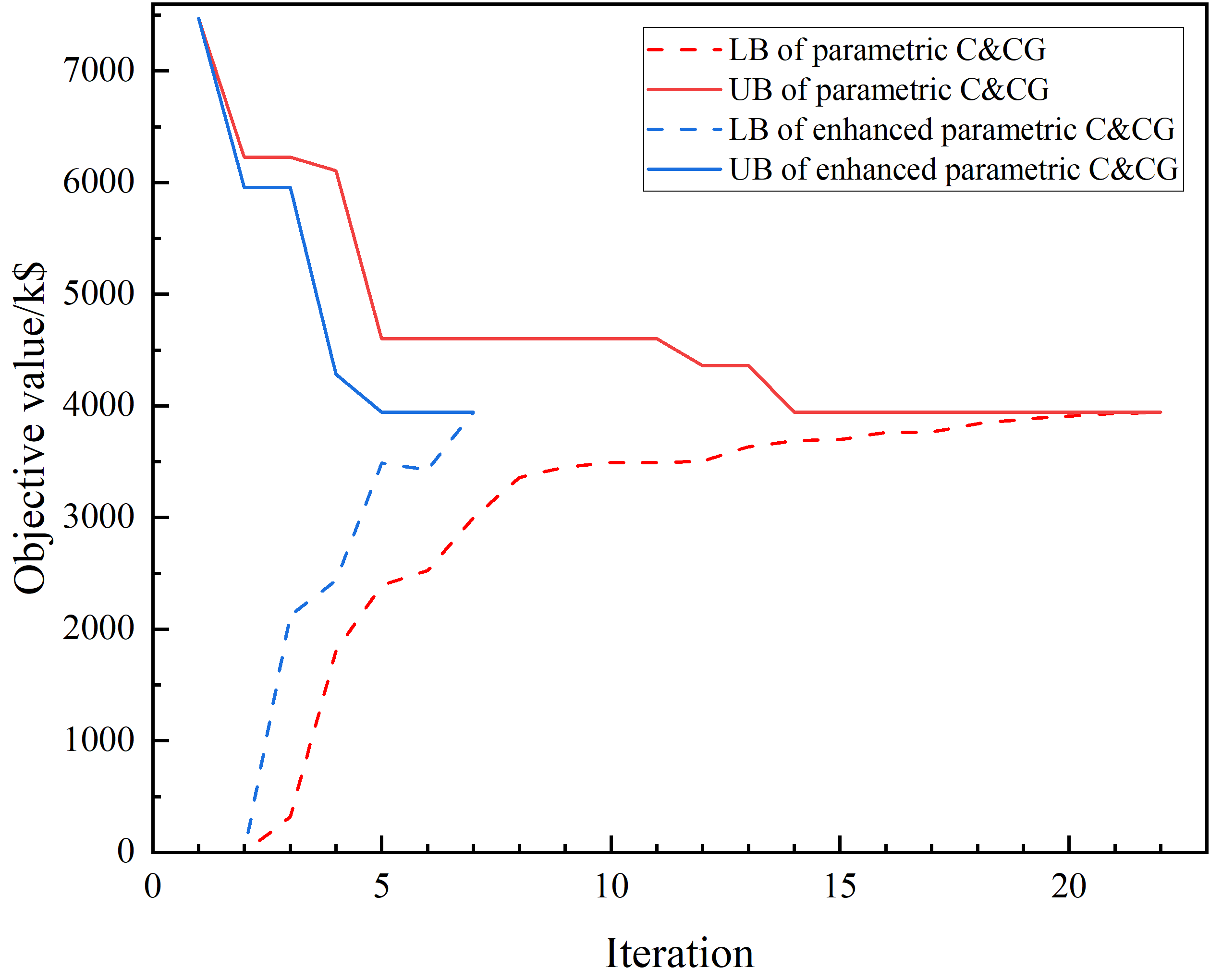}
	\caption{Iterative procedures of P-C\&CG algorithm with/without enhancement ($k=6$ and $\overline{\Upsilon}=8$)}
	\label{fig:IP} 
\end{figure}

\subsection{118-Bus Test Distribution Network}
The scalability of the proposed hardening method is tested on a 118-bus distribution network. It is equipped with 8 DGs and 7 ESSs, with an identical power rating as 500 kW. Also, the nominal storage capacity of each ESS unit is 1500 kWh, with a random fluctuation range of $\pm$20\%. Given $\overline{\Upsilon}=8$, $k^L=8$, and $k^{DG}=3$, the hardening scheme and associated worst-case contingency scenario is presented in Fig. \ref{fig:HS118}. Note that the blue shaded region contains the vulnerable distribution lines. The DGs located on nodes N9, N24, N50 and N57 could also be affected by extreme weather events. 

As shown in Fig. \ref{fig:HS118}, the worst-case contingency scenario is identified as the outages of L1-2, L2-10, L10-11, L11-18, L18-19, L23-24, L24-25 and L29-55, along with the DGs located at N10. The entire system will be divided into 9 parts, including 4 electrical islands and 4 dynamic microgrids (powered by DGs and ESSs). In such a case, the hardening decisions are made to protect lines L1-63, L63-64, L1-100, L100-101 and L100-114, which preserves the connectivity of substation node (i.e., N1) to a part of the critical nodes. Besides, due to the elevating protection for DGs at N24, N50 and N57, only the DGs at N10 suffers from components failure. The hardened DGs facilitates the on-emergency power supply of 3 microgrids, which expands the restoration coverage. As a consequence, the ratio of load shedding is controlled as 31.1\%.

Moreover, the solution tests are performed under different combinations of $k^L$ and $\overline{\Upsilon}$, by fixing $k^{DG}$ to 1. The enhanced P-C\&CG algorithm is compared with basic C\&CG and modified BD. As shown in Table \ref{tab:OS2}, the modified BD, in most of the test cases, fails to converge within the time limit ``T'' (i.e., 12 hours). In contrast, both the basic C\&CG and enhanced P-C\&CG can obtain the global optimum within 3850s, demonstrating the superior performance of C\&CGs for large-scale test systems. Additionally, in comparison to the basic C\&CG, the proposed P-C\&CG algorithm is more effective, reducing the average solution time from 1017.6s to 459.9s. For instance, when $k^L=8$ and $\overline{\Upsilon}=8$, the basic C\&CG requires 43 iterations and 1995.2s to converge. In contrast, by applying the enhanced P-C\&CG can reduce the iteration number to 14. The corresponding  solution time drops to 265.3s (reducing by 86.7\%). The aforementioned results have demonstrated the scalable solution capability of our customized algorithm, which significantly outperforms the existing solution measures.

\begin{figure*}[!t]
	\centering
	\includegraphics[width=5in]{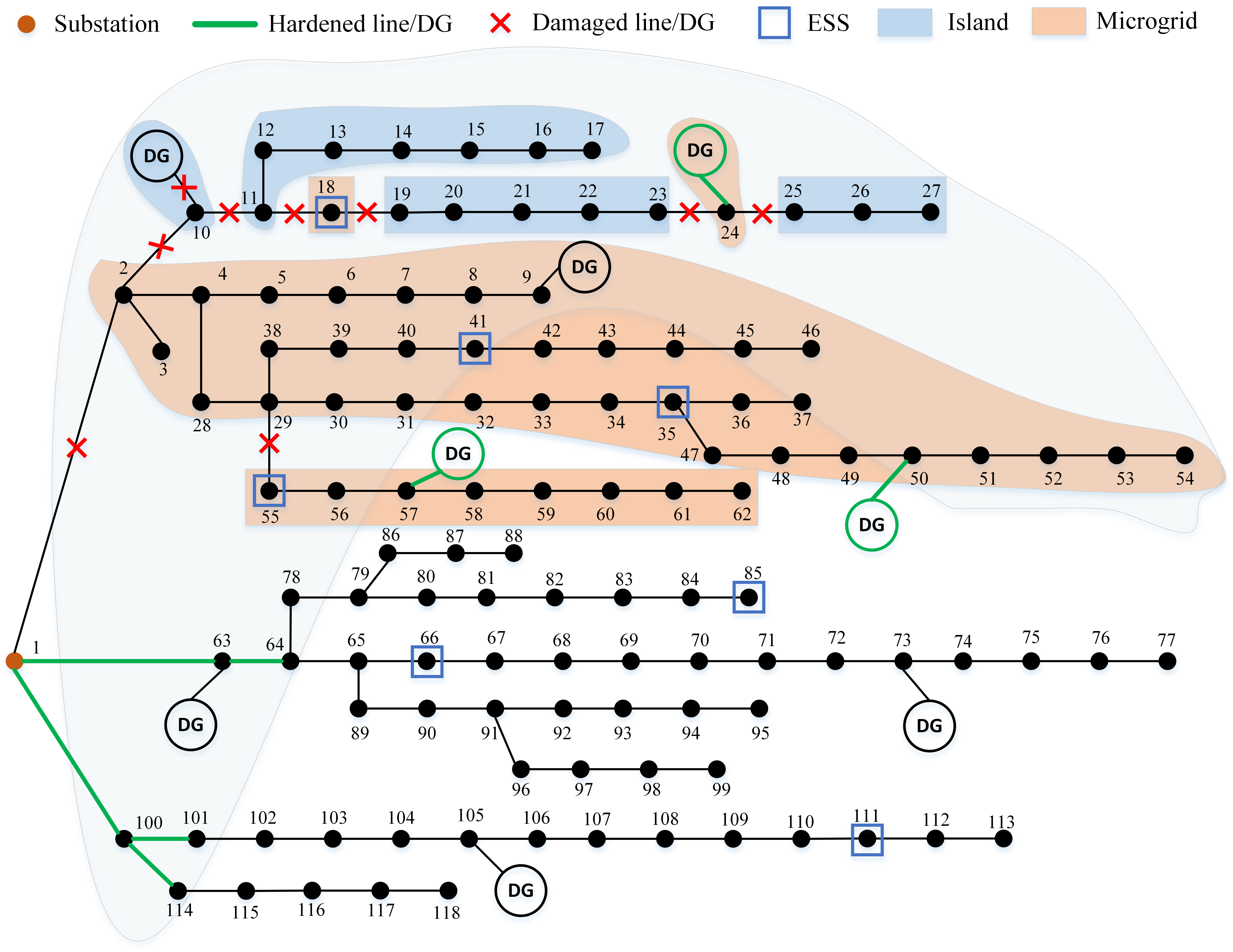}
	\caption{Resilient Hardening Results ($\overline{\Upsilon}=8$, $k^L=8$, and $k^{DG}=3$): 118-Bus Test System}
	\label{fig:HS118} 
\end{figure*}

\renewcommand\arraystretch{1.0}
\begin{table}[!t]
	\centering
	\caption{Computational Performance of Different Solution Methods: 118-Bus Test System}
	\setlength{\tabcolsep}{1.4mm}
	\scalebox{1}{
		\begin{tabular}{ccccccccc}
			\toprule
			\multirow{2}{*}{$k^L$} & \multirow{2}{*}{$\overline{\Upsilon}$}  & \multirow{2}{*}{OBJ} & \multicolumn{2}{c}{Enhanced P-C\&CG} & \multicolumn{2}{c}{Basic C\&CG} & \multicolumn{2}{c}{Modified BD}  \\
			&    &   & N\_itr & Time/s   & N\_itr & Time/s  & N\_itr & Time/s   \\
			\midrule
			\multirow{3}{*}{6} & 6 & 30.2\%  & 14     & 277.6   & 24    & 580.9     & 293    & 6394.9    \\
			& 8    & 28.0\%   & 17   & 459.8     & 38       & 1532.9      & /      & T        \\
			& 10    & 19.9\%     &9   &79.9     & 11      & 105.2    & /    & T     \\
			\midrule
			\multirow{3}{*}{8} & 6   & 31.6\%   &17   &524.8     & 21     & 452.1     & /      & T       \\
			& 8    & 29.5\%  &14   &265.3        & 43      & 1995.2      & /      & T       \\
			& 10   & 21.4\%  &8   &66.2     & 10      & 79.7     & /      & T
			\\
			\midrule
			\multirow{3}{*}{10} & 6  & 32.5\%   & 10   &218.6     & 21      & 458.1      & /      & T       \\
			& 8  &  30.7\%  & 20   &2170.4      & 58      & 3847.2      & /     & T          \\
			& 10  & 22.8\%   &8  &76.6       & 11      & 107.5     & /      & T        \\
			\bottomrule
	\end{tabular}}
	\label{tab:OS2}
\end{table} 

\section{Conclusion and Future Direction}
\label{Conclude}
This paper presents a DDU-based two-stage RSO formulation to support the proactive hardening decisions of distribution network and DGs against extreme events. A closed-form DDU set is constructed to analytically represent the endogenous uncertainty of system contingencies. Also, the exogenous uncertainties of load variations and available storage levels are considered through scenario-based SP in the innermost level problem. To address the computational challenges brought by a dynamically-shaped uncertainty set and the complicated trilevel stochastic structure, we develop and customize the P-C\&CG algorithm with nontrivial enhancements. Numerical studies on 33-bus and 118-bus test power distribution networks demonstrate the effectiveness and scalable capability of our resilient hardening method. Some key observations and insights are presented as in the following:

\begin{enumerate}
	\item \emph{Modeling Capability of DDU-RSO}: The DDU-RSO model provides a comprehensive vision on both the endogenous and exogenous uncertainties, which are coupled with different decision stages (i.e., the proactive hardening and post-event corrective operations). This new formulation can accurately characterize the impact of hardening actions on system contingencies, which enables robustly feasible decisions for resilience enhancement.
	\item \emph{Superiority of Enhanced P-C\&CG over Existing Solution Techniques}:  The customized P-C\&CG provides an exact and scalable solution tool to address the DDU-based two-stage RSO problem. Compared to the basic C\&CG and Benders-type dual decomposition method, the proposed algorithm demonstrates a superior solution capability, which has reduced the computation time by orders of magnitudes. 
	\item \emph{Significance of Exploiting ``Deep-Knowledge''}:  Our algorithm design can better utilize the deep knowledge of the underlying system (e.g., the inclusion of resilience importance indices) to achieve a strong scalable capacity. Such idea can be applied to investigate similar hardening issues among critical infrastructure systems.
\end{enumerate}

In our future work, the proposed method will be extended to achieve a greater practicality and solution adaptability. Note that our formulation is compatible for more detailed contingency models with respect to different types of catastrophic events. Additional resilience enhancement measures (e.g., the on-emergency demand response, network reconfiguration, and mobile energy resources) can be included. Certainly, the modeling of reconfiguration and mobile resources scheduling may introduce binary recourse variables, which converts the innermost level problem into an MIP. Although the classical RO model with MIP recourse can be exactly computed using the nested C\&CG algorithm \cite{zhao2012exact}, it cannot be directly applied to a novel DDU-RSO formulation. The future enhancement of P-C\&CG will focus on the efficient solution of DDU-RSO with mixed-integer recourse decisions.

\ifCLASSOPTIONcaptionsoff
  \newpage
\fi

\bibliographystyle{IEEEtran}
\bibliography{Ref}

\end{document}